\documentclass[letterpaper, 10 pt, conference]{ieeeconf}
\IEEEoverridecommandlockouts                              

\usepackage{amsmath,amssymb}
\usepackage{latexsym}
\usepackage{mathrsfs}
\usepackage{dsfont}
\usepackage{bm}
\usepackage{cancel}
\usepackage{array}
\usepackage{comment}
\usepackage{float}
\usepackage{graphicx}
\usepackage{epstopdf}
\usepackage{color}
\usepackage{enumerate}
\usepackage{mathtools}
\mathtoolsset{showonlyrefs}
\DeclareGraphicsExtensions{.eps}

\newcolumntype{H}{>{\setbox0=\hbox\bgroup}c<{\egroup}@{}}

\newcommand{\MCA}{\mathcal{A}}

\newcommand{\MCO}{\mathcal{O}}
\newcommand{\MCD}{\mathcal{D}}

\newcommand{\MCL}{\mathcal{L}}

\newcommand{\EE}{\mathbb{E}}

\newcommand{\RR}{\mathbb{R}}

\newcommand{\Ltx}{\mathcal{L}_{t,x}}

\newcommand{\Vy}{V^{\eps}}

\newcommand{\rz}{R}
\newcommand{\vz}{v^{(0)}}
\newcommand{\vo}{v^{(1)}}
\newcommand{\vt}{v^{(2)}}

\newcommand{\pz}{{\pi^{(0)}}}

\newcommand{\vzl}{v^{\pz,(0)}}
\newcommand{\vol}{v^{\pz,(1)}}
\newcommand{\vtl}{v^{\pz,(2)}}
\newcommand{\vthl}{v^{\pz,(3)}}

\newcommand{\Vyl}{V^{\pz,\eps}}

\newcommand{\pzt}{\widetilde\pi^0}
\newcommand{\pot}{\widetilde\pi^1}

\newcommand{\Vyt}{\widetilde{V}^{\eps}}

\newcommand{\vzt}{\widetilde{v}^{(0)}}
\newcommand{\vot}{\widetilde{v}^{(1)}}
\newcommand{\vtt}{\widetilde{v}^{(2)}}
\newcommand{\vtht}{\widetilde{v}^{(3)}}
\newcommand{\vat}{\widetilde{v}^{1\alpha}}
\newcommand{\vtat}{\widetilde{v}^{2\alpha}}

\newcommand{\avel}{\average{\lambda^2}}
\newcommand{\aves}{\overline\lambda}

\newcommand{\eps}{\epsilon}

\newcommand{\half}{\frac{1}{2}}

\newcommand{\abs}[1]{\left|#1\right|}
\newcommand{\average}[1]{\left\langle#1\right\rangle}

\newcommand{\ud}{\,\mathrm{d}}

\newtheorem{theo}{Theorem}[section]
\newtheorem{lem}[theo]{Lemma}

\newtheorem{rem}[theo]{Remark}
\newtheorem{prop}[theo]{Proposition}
\newtheorem{assump}[theo]{Assumption}
\newtheorem{cor}[theo]{Corollary}

\begin{document}

\title{\LARGE \bf Asymptotic Optimal Portfolio in Fast Mean-reverting Stochastic Environments}
\author{ Ruimeng Hu\thanks{Ruimeng Hu is with the Department of Statistics,
	Columbia University,
	New York, NY 10027, USA {\tt\small rh2937@columbida.edu}. The work was mainly done when RH was a graduate student at the University of California, Santa Barbara, and was partially supported by the National Science Foundation
under DMS-1409434.}
}
\maketitle
\thispagestyle{empty}
\pagestyle{empty}

\begin{abstract}
	
This paper studies the portfolio optimization problem when the investor's utility is general and the return and volatility of the risky asset are fast mean-reverting, which are important to capture the fast-time scale in the modeling of stock price volatility. Motivated by the heuristic derivation in [J.-P. Fouque, R. Sircar and T. Zariphopoulou, \emph{Mathematical Finance}, 2016], we propose a zeroth order strategy, and show its asymptotic optimality within a specific (smaller) family of admissible strategies under proper assumptions. This optimality result is achieved by establishing a first order approximation of the problem value associated to this proposed strategy using singular perturbation method, and estimating the risk-tolerance functions. The results are natural extensions of our previous work on portfolio
optimization in a slowly varying stochastic environment [J.-P. Fouque and R. Hu, \emph{SIAM Journal
on Control and Optimization}, 2017], and together they form a whole picture of analyzing portfolio optimization in both fast and slow environments.

\end{abstract}


\section{Introduction}\label{sec_intro}
The portfolio optimization problem in continuous time, also known as the Merton problem, was firstly studied in \cite{Me:69,Me:71}. In his original work, explicit solutions on how to allocate money between risky and risk-less assets and/or how to consume wealth are provided so that the investor's expected utility is maximized, when the risky assets follows the Black--Scholes (BS) model and the utility is of Constant Relative Risk Aversion (CRRA) type. Since these seminal works, lots of research has been done to relax the original model assumptions, for example, to allow transaction cost \cite{MaCo:76}, \cite{GuMu:13}, drawdown constraints \cite{GrZh:93},  \cite{CvKa:95},  \cite{ElTo:08},  price impact \cite{CuCv:98}, and stochastic volatility \cite{Za:99}, \cite{ChVi:05}, \cite{FoSiZa:13} and \cite{LoSi:16}.

Our work extends Merton's model by allowing more general utility, and by modeling the return and volatility of the risky asset $S_t$ by a fast mean-reverting process $Y_t$:
\begin{align}
&\ud S_t = \mu(Y_t)S_t\ud t + \sigma(Y_t)S_t\ud W_t,\label{eq_Stofy}\\
&\ud Y_t = \frac{1}{\eps} b(Y_t)\ud t + \frac{1}{\sqrt\eps} a(Y_t) \ud W_t^Y\label{eq_Yt}.
\end{align}
The two standard Brownian motion (Bm) are imperfectly correlated: $\ud \average{W, W^Y} = \rho \ud t, \; \rho \in (-1,1)$. We are interested in the terminal utility maximization problem
\begin{equation}\label{def_Vy}
\Vy(t,x,y) \equiv \sup_{\pi \in \MCA^\eps}\EE[U(X_T^\pi)\vert X_t^\pi = x, Y_t = y],
\end{equation}
where $X_t^\pi$ is the wealth associated to self-financing $\pi$:
\begin{equation}\label{def_Xt}
\ud X_t^\pi = \pi(t,X_t^\pi,Y_t)\mu(Y_t)\ud t + \pi(t,X_t^\pi, Y_t)\sigma(Y_t) \ud W_t,
\end{equation}
(assume the risk-free interest rate varnishes $r = 0$) and $\MCA^\eps$ is the set of strategies that $X_t^\pi$ stays nonnegative. Using singular perturbation technique, our work provide an \emph{asymptotic optimal strategy} $\pz$ within a specific class of admissible strategies $\MCA_0^\eps$ that satisfies certain assumptions:
\begin{equation}\label{def_class}
\MCA_0^\eps\left[\pzt, \pot,\alpha\right] = \left\{ \pzt + \eps^\alpha \pot\right\}_{0 \leq \eps \leq 1}.
\end{equation}

\textbf{Motivation and Related Literature.} The reason to study the proposed problem is threefold. Firstly, in the direction of asset modeling \eqref{eq_Stofy}-\eqref{eq_Yt}, the well-known implied volatility smile/smirk phenomenon leads us to employ a BS-like stochastic volatility model. Empirical studies have identified scales in stock price volatility: both fast-time scale on the order of days and slow-scale on the order of months \cite{FoPaSiSo:11}. This results in putting a parameter $\eps$ in \eqref{eq_Yt}. The slow-scale case (corresponding to large $\eps$ in \eqref{eq_Yt}), which is particularly important in long-term investments, has been studied in our previous work \cite{FoHu:16}.  An asymptotic optimality strategy is proposed therein using regular perturbation techniques. This makes it natural to extend the study to fast-varying regime, where one needs to use singular perturbation techniques. Secondly, in the direction of utility modeling, apparently not everyone's utility is of CRRA type \cite{BrNa:08}, therefore it is important to consider to work under more general utility functions. Thirdly, although it is natural to consider multiscale factor models for risky assets, with a slow factor and a fast factor as in \cite{FoSiZa:13}, more involved technical calculation and proof are required in combining them, and thus, we leave it to another paper in preparation \cite{Hu:XX}.

Our proposed strategy $\pz$ is motivated by the heuristic derivation in \cite{FoSiZa:13}, where a singular perturbation is performed to the PDE satisfied by $\Vy$. This gave a formal approximation $\Vy = \vz + \sqrt\eps{\vo} + \eps \vt + \cdots$. They then conjectured that the zeroth order strategy
\begin{equation}\label{eq_pz}
\pz(t,x,y)= -\frac{\lambda(y)}{\sigma(y)}\frac{\vz_x(t,x,y)}{\vz_{xx}(t,x,y)}, \quad \lambda(y) = \frac{\mu(y)}{\sigma(y)}
\end{equation}
reproduces the optimal value up to the first order $\vz  + \sqrt{\eps}\vo$, with $\vz$ and $\vo$ given by \eqref{eq_vz} and \eqref{eq_voofy}.

\textbf{Main Theorem.} Let $\Vyl$ (resp. $\Vyt$) be the expected utility of terminal wealth associated to $\pz$ (resp. $\pi \in \MCA_0^\eps$): $\Vyl := \EE[U(X_t^\pz) \vert X_t^\pz = x, Y_t = y],$
 and $X_t^\pz$ be the wealth process given by \eqref{def_Xt} with $\pi = \pz$ (resp. $\pi$ in $\MCA_0^\eps$). By comparing $\Vyl$ and $\Vyt$, we claim that $\pz$ performs asymptotically better up to order $\sqrt\eps$
 than the family $\left\{\pzt+\eps^\alpha\pot\right\}$. Mathematically, this is formulated as:

\begin{theo}\label{Thm_optimality}
	Under assumptions detailed in Sections~\ref{sec:assumptions} and \ref{sec_asymptoticoptimality},
	for any family of trading strategies $\MCA_0^\eps\left[ \pzt, \pot,\alpha\right] = \{\pzt + \eps^\alpha\pot\}_{0 \leq \eps\leq 1}$, the following limit exists and satisfies
	\begin{equation}\label{eq_Vztineq}
	\ell := \lim_{\eps\to 0}\bigl(\Vyt(t, x, y)-\Vyl(t, x, y))/\sqrt{\eps}\leq 0.
	\end{equation}
\end{theo}
Proof will be given in Section~\ref{sec_asymptoticoptimality} as well as the interpretations of this inequality according to different $\alpha$'s. Our main theorem gives some insights on how to construct expansion of the optimal $\pi$, which, however, is still an open question. We remark that, in a related work \cite{Pa:13}, expansion results for $\pi^\ast$ exist under a discrete-time filtering setting.

The rest of the paper is organized as follows. Section~\ref{sec:assumptions} introduces some preliminaries of the Merton problem and standing assumptions in this paper. Section~\ref{sec_fast}  gives $\Vyl$'s first order approximation $\vz + \sqrt\eps{\vo}$. Section~\ref{sec_asymptoticoptimality} is dedicated to the proof of Theorem~\ref{Thm_optimality}. The expansion of $\Vyt$ is analyzed first, with precise derivations, while the detailed technical assumption is referred to our recent work \cite{FoHu:16}. 


\section{Preliminaries and Assumptions}\label{sec:assumptions}

In this section, we firstly review the classical Merton problem, and the notation of risk tolerance function $R(t,x;\lambda)$. Then heuristic expansion results of $\Vy$ in \cite{FoSiZa:13} are summarized. Standing assumptions of this paper are listed, as well as some estimations regarding $R(t,x;\lambda)$ and $\vz$. We refer to our recent work \cite[Section~2, 3]{FoHu:16} for proofs of all these results.

We shall first consider the case of constant $\mu$ and $\sigma$ in \eqref{eq_Stofy}. This is the classical Merton problem, which plays a crucial role in interpreting the leading order term $\vz$  and analyzing the singular perturbation. This problem has been studied intensively, for instance, in \cite{KaSh:98}. Let $X_t$ be the wealth process in this case. Using the notation in \cite{FoSiZa:13}, we denote by $M(t,x;\lambda)$ the problem value. In Merton's original work, closed-form $M(t,x;\lambda)$ was obtained when the utility $U(\cdot)$ is of power type. In general, one has the following results, with proofs given in \cite[Section~2.1]{FoHu:16} or the references therein.

%
%
\begin{prop}\label{prop_Merton}
	Assume that the utility function $U(x)$ is $C^2(0,\infty)$, strictly increasing, strictly concave, such that $U(0+)$ is finite, and satisfies the Inada and Asymptotic Elasticity conditions:
$U'(0+) = \infty$, $U'(\infty) = 0$, $\text{AE}[U] := \lim_{x\rightarrow \infty} x\frac{U'(x)}{U(x)} <1$, then, the Merton value function $M(t,x;\lambda)$ is strictly increasing, strictly concave in the wealth variable $x$, and decreasing in the time variable $t$. It is $C^{1,2}([0,T]\times \RR^+)$ and is the unique solution to the Hamilton-Jacobi-Bellman(HJB) equation, with $M(T,x;\lambda) = U(x)$,
	\begin{equation}\label{eq_value}
	 M_t+\sup_{\pi}\left\{\frac{\sigma^2}{2}\pi^2M_{xx}+\mu\pi M_x\right\}
	= 0,
	\end{equation}
	where $\lambda = \frac{\mu}{\sigma}$ is the Sharpe ratio. It is  $C^1$ w.r.t $\lambda$, and the optimal strategy is given by $\pi^\star(t,x;\lambda)=-\frac{\lambda}{\sigma}\frac{M_x(t,x;\lambda)}{M_{xx}(t,x;\lambda)}.$
\end{prop}


We next define the risk-tolerance function $R(t,x;\lambda) = -\frac{M_x(t,x;\lambda)}{M_{xx}(t,x,;\lambda)}$,
and operators following the notations in \cite{FoSiZa:13},
\begin{equation}\label{def_dkltx}
D_k = R(t,x;\lambda)^k \partial_x^k, 
\quad \Ltx(\lambda) = \partial_t + \frac{1}{2}\lambda^2D_2 + \lambda^2D_1.
\end{equation}
By the concavity of $M(t,x;\lambda)$, $R(t,x;\lambda)$ is continuous and strictly positive. Using the relation $D_1 M = - D_2 M$, the nonlinear Merton PDE \eqref{eq_value} can be re-written in a ``linear'' way: $\Ltx(\lambda)M(t,x;\lambda) = 0$. We now mention a uniqueness result to this PDE, which will be used repeatedly in Sections~\ref{sec_fast}.
\begin{prop}\label{prop_ltxunique}
	Let $\Ltx(\lambda)$ be the operator defined in \eqref{def_dkltx}, and assume that the utility function $U(x)$ satisfies the conditions in Proposition \ref{prop_Merton}, then
	\begin{equation}\label{def_ltxpde}
	 \Ltx(\lambda)u(t,x;\lambda) = 0,\quad
	u(T,x;\lambda) = U(x),
	\end{equation}
	has a unique nonnegative solution.
\end{prop}


Next, we review the formal expansion results of $\Vy$ derived in \cite{FoSiZa:13}. To apply singular perturbation technique, we assume that the process $Y^{(1)}_t \stackrel{\MCD}{=} Y_{t\eps}$ is ergodic and equipped with a unique invariant distribution $\Phi$. We use the notation $\average{\cdot}$ for averaging w.r.t. $\Phi$, namely, $\average{f} = \int f \ud \Phi$. Let $\MCL_0$ be the infinitesimal generator of $Y^{(1)}$: $\MCL_0 = \frac{1}{2}a^2(y)\partial_{y}^2 + b(y)\partial_y.$
%
%
%
%
Then, by dynamic programming principle, the value function $\Vy$ solves the HJB equation in the viscosity sense:
\begin{multline}\label{eq_Vy}
\Vy_t + \frac{1}{\eps}\MCL_0\Vy + \max_{\pi \in \MCA^\eps}\Bigl(\sigma(y)^2\pi^2\Vy_{xx}/2 \\ + \pi\left(\mu(y)\Vy_x +  \rho  a(y)\sigma(y) \Vy_{xy}/\sqrt{\eps}\right)\Bigr) = 0.
\end{multline}
and its regularity is not clear. In \cite{FoSiZa:13}, a unique classical solution is assumed in order to perform heuristic derivations. Moreover, the optimizer in \eqref{eq_Vy} is well-defined:
		$\pi^\ast = -\frac{\lambda(y)\Vy_x}{\sigma(y)\Vy_{xx}} - \frac{\rho  a(y)\Vy_{xy}}{\sqrt{\eps}\sigma(y)\Vy_{xx}},$
	and the simplified HJB equation reads:
	$\Vy_t + \frac{1}{\eps} \MCL_0\Vy - \left(\lambda(y)\Vy_x + \frac{1}{\sqrt{\eps}}\rho  a(y)\Vy_{xy}\right)^2/(2\Vy_{xx}) = 0,$
	for $(t,x,y) \in [0,T] \times \RR^+ \times \RR$.  We remark that, to obtain our Theorem~\eqref{Thm_optimality}, the smooth condition is not needed, as we focus on the quantity $\Vyl$ defined in \eqref{def_Vyl}. It corresponds to a linear PDE, for which classical solutions exist.
	
The equation \eqref{eq_Vy} is fully nonlinear and is only explicitly solvable in some cases; see \cite{ChVi:05} for instance. The heuristic expansions overcome this by providing approximations to $\Vy$.  This is done by the so-called singular perturbation method, as often seen in homogenization theory. To be specific, one substitutes the expansion $\Vy = \vz + \sqrt\eps\vo + \eps\vt + \cdots$ into the above equation, establishes equations about $v^{(k)}$ by collecting terms of different orders. In \cite[Section 2]{FoSiZa:13}, this is performed for $k = 0, 1$ and we list their results as follows:
\begin{enumerate}[(i)]
	\item The \emph{leading order term} $\vz(t,x)$ is defined as the solution to the Merton PDE associated with the averaged Sharpe ratio $\overline\lambda = \sqrt{\avel}$:
	\begin{equation}
	\vz_t - \frac{1}{2}\overline\lambda^2\frac{\left(\vz_x\right)^2}{\vz_{xx}} = 0, \quad \vz(T,x) = U(x),\label{eq_vz}
	\end{equation}
	and by Proposition~\ref{prop_Merton} $\vz$ is identified as:
	\begin{equation}\label{eq_vzandmerton}
	\vz(t,x) = M\bigl(t,x;\overline\lambda\bigr).
	\end{equation}
	\item The \emph{first order correction} $\vo$ is identified as the solution to the linear PDE:
	\begin{equation}\label{eq_voofy}
	\vo_t + \frac{\overline\lambda^2}{2}(\frac{\vz_x}{\vz_{xx}})^2\vo_{xx} - \overline\lambda^2 \frac{\vz_x}{\vz_{xx}}\vo_x = \frac{\rho}{2} BD_1^2\vz,
	\end{equation}
	with $\vo(T,x)= 0$. The constant
	$B = \average{\lambda a \theta'}$, and $\theta(y)$ solves $\MCL_0 \theta(y) = \lambda^2(y) - \overline{\lambda}^2.$
	Rewrite equation \eqref{eq_voofy} in terms of the operators in \eqref{def_dkltx}, $\vo$ solves the following PDE which admits a unique solution:
	\begin{equation}\label{eq_vo}
	\Ltx(\overline\lambda)\vo = \half\rho BD_1^2\vz, \quad \vo(T,x) = 0.
	\end{equation}
	\item $\vo$ is explicitly given in term of $\vz$ by
	$\vo(t,x) = -\frac{1}{2}(T-t)\rho BD_1^2\vz(t,x).$
\end{enumerate}

Now we introduce the assumptions on the utility $U(\cdot)$ and the state processes $(S_t, X_t^\pz, Y_t)$, and refer to \cite[Section~2]{FoHu:16} for further discussions and remarks.



\begin{assump}\label{assump_U}
Throughout the paper, we make the following assumptions on the utility $U(x)$:
\begin{enumerate}[(i)]
  \item\label{assump_Uregularity}  $U(x)$ is $C^7(0,\infty)$, strictly increasing, strictly concave and satisfying the following conditions:
$U'(0+) = \infty, \quad U'(\infty) = 0, \quad \text{AE}[U] := \lim_{x\rightarrow \infty} x\frac{U'(x)}{U(x)} <1.$
\item\label{assump_Ubddbelow}$U(0+)$ is finite. Without loss of generality, $U(0+) = 0$.
  \item\label{assump_Urisktolerance} Denote by $R(x)$ the risk tolerance,
  $R(x) := -\frac{U'(x)}{U''(x)}.$
Assume that $R(0) = 0$, $R(x)$ is strictly increasing and $R'(x) < \infty$ on $[0,\infty)$, and there exists $K\in\RR^+$, such that for $x \geq 0$, and $ 2\leq i \leq 5$,
      \begin{equation}\label{assump_Uiii}
      \abs{\partial_x^iR^i(x)} \leq K.
      \end{equation}
  \item\label{assump_Ugrowth} Define the inverse function of the marginal utility $U'(x)$ as $I: \RR^+ \to \RR^+$, $I(y) = U'^{(-1)}(y)$, and assume that, for some positive $\alpha$, $I(y)$ satisfies the polynomial growth condition:
$I(y) \leq \alpha + \kappa y^{-\alpha}.$
\end{enumerate}
\end{assump}

Assumption~\ref{assump_U}\eqref{assump_Ubddbelow} is a sufficient condition, and rules out the cases $U(x) = \frac{x^{\gamma}}{\gamma}$, $\gamma < 0$, and $U(x) = \log(x)$. However, all theorems in the paper still hold, as it is to ensure that terms in \eqref{eq_Eylprob} are of the form \eqref{eq_opD}, which is automatically satisfied for aforementioned cases. Next are the model assumptions.


\begin{assump}\label{assump_valuefunc} We make the following assumptions on the state processes $(S_t, X_t^\pz, Y_t)$:
\begin{enumerate}[(i)]
	\item\label{assump_valuefuncSZ} For any starting points $(s, y)$ and fixed $\eps$, the system of SDEs \eqref{eq_Stofy}--\eqref{eq_Yt} has a unique strong solution $(S_t, Y_t)$. The functions $\lambda(y)$ and $a(y)$ have polynomial growth.
	
	\item\label{assump_valuefuncZmoment} The process $Y^{(1)}$ with infinitesimal generator $\MCL_0$ is ergodic with a unique invariant distribution, and admits moments of any order uniformly in $t \leq  T$:
	$\sup_{t \leq T}\left\{ \EE\abs{Y_t^{(1)}}^k\right\} \leq C(T,k)$. The solution $\phi(y)$ of the Poisson equation $\MCL_0 \phi = g$ is assumed to be polynomial for polynomial functions $g$.
	
%
         \item\label{assump_valuefuncX} The wealth process $X_\cdot^\pz$ is in  $L^2([0,T]\times \Omega)$ uniformly in $\eps$ , i.e.,
             $\EE_{(0,x,y)}\left[\int_0^T \left(X_s^\pz\right)^2 \ud s\right] \leq C_2(T,x,y)$,
             where $C_2(T,x,y)$ is independent of $\eps$ and $\EE_{(0,x,y)}[\cdot] = \EE[\cdot\vert X_0 = x, Y_0 = y]$.

       \end{enumerate}
\end{assump}

%
%


Here we provide several estimations of the risk tolerance function $R(t,x;\lambda)$ and the zeroth order value function $\vz$, which are crucial in the proof of Theorem \ref{Thm_one}. 

By Proposition~\ref{prop_Merton} and the relation \eqref{eq_vzandmerton}, $\vz$ is concave in the wealth variable $x$, and decreasing in the time variable $t$, therefore has a linear upper bound, for $(t,x) \in [0,T]\times \RR^+$:
$\vz(t,x) \leq \vz(0,x) \leq c + x$, for some constant $c$.
Combining it with Assumption~\ref{assump_valuefunc}\eqref{assump_valuefuncX}, we deduce:

\begin{lem}\label{lem_unibdd}
	Under Assumption~\ref{assump_U} and \ref{assump_valuefunc}, the process $\vz(\cdot, X_\cdot^\pz)$ is in $L^2([0,T]\times \Omega)$ uniformly in $\eps$, i.e. $\forall (t,x) \in [0,T]\times\RR^+$:
	$\EE_{(t,x)}\left[\int_t^T \left(\vz(s,X_s^\pz)\right)^2 \ud s\right] \leq C_3(T,x),$
	where $\vz(t,x)$ satisfies equation \eqref{eq_vz}.
\end{lem}

\begin{prop}\label{prop_mono}
Suppose the risk tolerance $R(x) = -\frac{U'(x)}{U''(x)}$ is strictly increasing for all $x$ in $[0,\infty)$ (this is part of Assumption \ref{assump_U} \eqref{assump_Urisktolerance}), then, for each $t \in [0,T)$,  $R(t,x;\lambda)$ is strictly increasing in the wealth variable $x$.
\end{prop}

\begin{prop}\label{prop_ssh}
Under Assumption \ref{assump_U}, the risk tolerance function $R(t,x; \lambda)$ satisfies: $\forall 0 \leq j \leq 4$, $\exists K_j >0$, such that $\forall (t,x) \in [0,T)\times \RR^+ $, $\abs{R^j(t,x;\lambda)\left(\partial_x^{j+1}R(t,x;\lambda)\right)} \leq K_j.$
Or equivalently, $\forall 1 \leq j \leq 5$, there exists $\widetilde K_j>0$, such that
$\abs{\partial_x^j R^j(t,x;\lambda)} \leq \widetilde K_j.$
Moreover,
one has $R(t,x;\lambda) \leq K_0x.$
\end{prop}

%

\section{Portfolio performance of a given strategy}\label{sec_fast}
Recall the strategy $\pz$ defined in \eqref{eq_pz},
and assume $\pz$ is admissible. In this section, we are interested in studying its performance. That is, to give approximation results of  the value function associated to $\pz$,  denote by $\Vyl$:
\begin{equation}\label{def_Vyl}
\Vyl(t,x,y) = \EE\left\{U(X_T^\pz) | X_t^\pz = x, Y_t = y\right\},
\end{equation}
where $U(\cdot)$ is a general utility function satisfying Assumption \ref{assump_U},  $X_t^\pz$ is the wealth process associated to the strategy $\pz$ and $Y_t$ is the fast factor.
Our main result of this section is the following, with the proof delayed in Section~\ref{sec_accproof}.

\begin{theo}\label{Thm_one}
	Under assumptions \ref{assump_U} and \ref{assump_valuefunc}, the residual function $E(t,x,y)$ defined by
	$E(t,x,y):= \Vyl(t,x,y) - \vz(t,x)-\sqrt\eps\vo(t,x),$
	is of order $\eps$.  In other words, $\forall (t,x,y) \in [0,T]\times\RR^+ \times \RR$, $E(t,x,y) \leq C\eps$, for some constant $C$ depending on $(t,x,y)$ but not on $\eps$.
\end{theo}


\begin{cor}\label{cor_optimality}
	In the case of power utility $U(x) = \frac{x^\gamma}{\gamma}$, $\pz$ is asymptotically optimal in $\MCA^\eps(t,x,y)$ up to order $\sqrt\eps$.
\end{cor}
\begin{proof}
	This is obtained by comparing expansions of $\Vy$ given in \cite[Corollary 6.8]{FoSiZa:13}, and of $\Vyl$ from the above Theorem. Since both quantities have the approximation $\vz + \sqrt{\eps}\vo$ at order $\sqrt{\eps}$, we have the desired result.
\end{proof}

\subsection{Formal expansion of $\Vyl$}\label{sec_formalexpansion}
In the following derivation, to condense the notation, we use $R$ for $R(t,x;\aves)$, and $\pz$ for $\pz(t,x,y)$ given in \eqref{eq_pz}.

By the martingale property, $\Vyl$ solves the linear PDE: $\Vyl_t + \pz\mu(y)\Vyl_x+ \frac{\pz}{\sqrt\eps}\rho a(y)\sigma(y)\Vyl_{xy}
+ \frac{1}{\eps}\MCL_0\Vyl+ \frac{1}{2}\sigma^2(y)\left(\pz\right)^2\Vyl_{xx} = 0$.
Define two operators $\MCL_1$ and $\MCL_2$ by $\MCL_1 = \rho a(y)\sigma(y)\pz\partial_{xy} = \rho a(y)\lambda(y)R(t,x;\aves)\partial_{xy}$, and $\MCL_2 = \partial_t + \frac{1}{2}\sigma^2(y)\left(\pz\right)^2\partial_x^2 + \mu(y)\pz\partial_x = \partial_t +\frac{1}{2}\lambda^2(y)D_2 + \lambda^2(y)D_1$ respectively,
then this linear PDE can be rewritten as:
\begin{equation}\label{eq_Vylop}
\left(\MCL_2 +  \MCL_1/{\sqrt\eps} + \MCL_0/{\eps}\right)\Vyl = 0.
\end{equation}

We look for an expansion of $\Vyl$ of the form
$\Vyl = \vzl + \sqrt\eps \vol + \eps \vtl + \cdots,$
with
$\vzl(T,x,y) = U(x)$ and $v^{\pz, (k)}(T,x,y) = 0$, for $k \geq 1.$
Inserting the above expansion of $\Vyl$ into \eqref{eq_Vylop}, and collecting terms of $\MCO(\frac{1}{\eps})$ and $\MCO(\frac{1}{\sqrt\eps})$ give:
$\MCL_0 \vzl = 0,\quad \MCL_0 \vol + \MCL_1\vzl = 0.$
Since $\MCL_0$ and $\MCL_1$ are operators taking derivatives in $y$, we make the choice that $\vzl$ and $\vol$ are \emph{independent of $y$}.  Next, collecting terms of $\MCO(1)$ yields
$\MCL_0 \vtl + \MCL_2\vzl = 0,$
whose solvability condition requires that $\average{\MCL_2 \vzl} = 0$. This leads to a PDE for $\vzl$:
$\vzl_t + \frac{1}{2}\overline\lambda^2\left(\rz\right)^2\vzl_{xx} + \overline\lambda^2\rz\vzl_x = 0$, $\vzl(T,x) = U(x),$
which has a unique solution (c.f. Proposition~\ref{prop_ltxunique}). Since $\vz$ also solves this equation, we deduce that
$\vzl(t,x) \equiv \vz(t,x) = M(t,x;\overline\lambda)$,
 and $\vtl$ admits a solution
$
\vtl(t,x,y) = -\half\theta(y) D_1\vz + C_1(t,x),
$ with $\theta(y)$ given by $\MCL_0 \theta(y) = \lambda^2(y) - \overline{\lambda}^2$
and $D_k$ in \eqref{def_dkltx}.

Then, collecting terms of order $\sqrt\eps$  yields
$\MCL_2\vol + \MCL_1\vtl + \MCL_0 \vthl = 0,$
and the solvability condition reads $\average{\MCL_2 \vol + \MCL_1 \vtl} = 0$. This gives an equation satisfied by $\vol$:
$\vol_t + \frac{1}{2}\overline\lambda^2\left(\rz\right)^2\vol_{xx} + \overline\lambda^2\rz\vol_x - \half\rho B D_1^2 \vz  = 0,
\quad \vol(T,x) = 0$,
which is exactly equation \eqref{eq_voofy}. This equation is uniquely solved by $\vo$ (see \eqref{eq_vo}). Thus, we obtain
$\vol \equiv \vo = -\frac{1}{2}(T-t)\rho B D_1^2\vz$.

Using the solution of $\vol$ and $\vtl$ we just identified, one deduces an expression for $\vthl$:
$\vthl = \frac{1}{2}(T-t)\theta(y)\rho B\left(\frac{1}{2}D_2 + D_1\right)D_1^2\vz + \frac{1}{2} \rho \theta_1(y)D_1^2\vz + C_2(t,x),$
where $\theta_1(y)$ is the solution to the Poisson equation:
$\MCL_0\theta_1(y) = a(y)\lambda(y)\theta'(y) - \average{a\lambda\theta'}.$

%

\subsection{First order accuracy: proof of Theorem~\ref{Thm_one}}\label{sec_accproof}
This section completes the proof of Theorem~\ref{Thm_one}, which shows the residual function $E(t,x,y)$ is of order $\eps$. To this end, we define the auxiliary residual function $\widetilde E(t,x,y)$ by
$\widetilde E = \Vyl - (\vz + \eps^{1/2}\vo + \eps \vtl + \eps^{3/2}\vthl),$
where we choose $C_1(t,x) = C_2(t,x) \equiv 0$ in the expression of $\vtl$ and $\vthl$. Then, it remains to show $\widetilde E \sim \eps$.

According to the derivation in Section~\ref{sec_formalexpansion}, the auxiliary residual function $\widetilde E$ solves
$\left(\frac{1}{\eps}\MCL_0 + \frac{1}{\sqrt\eps}\MCL_1 + \MCL_2\right)\widetilde E + \eps(\MCL_1\vthl + \MCL_2\vtl) + \eps^{3/2}\MCL_2 \vthl = 0,$
with a terminal condition
$\widetilde E(T,x,y) = -\eps\vtl(T,x,y) - \eps^{3/2}\vthl(T,x,y).$
Note that $\frac{1}{\eps}\MCL_0 + \frac{1}{\sqrt\eps}\MCL_1 + \MCL_2$ is the infinitesimal generator of the processes $\left(X_t^\pz, Y_t\right)$, one applies Feynman-Kac formula and deduces:
{\small\begin{align}
\widetilde E(t,x,y) &= \eps \EE_{(t,x,y)}\left[\int_t^T \MCL_1\vthl (s,X_s^\pz, Y_s) \ud s \right] \nonumber \\
&+ \eps \EE_{(t,x,y)}\left[\int_t^T \MCL_2\vtl  (s,X_s^\pz, Y_s) \ud s \right] \nonumber\\
& + \eps^{3/2} \EE_{(t,x,y)}\left[\int_t^T \MCL_2\vthl  (s,X_s^\pz, Y_s) \ud s \right] \nonumber \\
& - \eps\EE_{(t,x,y)}\left[\vtl(T,X_T^\pz, Y_T) \right] \nonumber\\
&  - \eps^{3/2}\EE_{(t,x,y)}\left[\vthl(T,X_T^\pz,Y_T)\right].\label{eq_Eylprob}
\end{align}}
The first three expectations come from the source terms while the last two come from the terminal condition. We shall prove that each expectation above is uniformly bounded in $\eps$. The idea is to relate them to the leading order term $\vz$ and the risk-tolerance function $R(t,x;\lambda)$, where some nice properties and estimates are already established in Section~\ref{sec:assumptions}.

For the source terms, straightforward but tedious computations give:
{\tiny
\begin{align}
&\MCL_2\vtl = -\frac{1}{4}\theta(y)\left(\lambda^2(y) - \overline\lambda^2\right)D_1^2 \vz,\\
&\MCL_1\vthl = \frac{1}{2}\rho^2a(y)\lambda(y)\theta_1'(y)D_1^3\vz +\\
&\quad \frac{1}{2}(T-t)\rho^2Ba(y)\lambda(y)\theta'(y)D_1\left[\frac{1}{2}D_2+D_1\right]D_1^2\vz,\\
&\MCL_2\vthl = \frac{1}{4}\rho\theta_1(y)\left(\lambda^2(y)-\overline\lambda^2\right)D_1^3\vz \\
&\quad + \frac{1}{2}\theta(y)\rho B\left\{-\left[\frac{1}{2}D_2 + D_1\right]D_1^2\vz + \frac{1}{2}(T-t)\left(\lambda^2(y) - \overline\lambda^2\right)D_1^4\vz\right\}  \\
&\quad + \frac{1}{4}\theta(y)\rho B(T-t)\\
& \quad \quad\times \left[\frac{1}{2}\left(\lambda^2(y) - \overline\lambda^2\right)D_2D_1^3\vz - \lambda^2(y)RR_{xx}(D_2+D_1)D_1^2\vz\right],
\end{align}}
where in the computation of $\MCL_2\vthl$, we use the commutator between operators $D_2$ and $\MCL_2$:
$[\MCL_2, D_2]w = \MCL_2D_2 w - D_2 \MCL_2 w = -\lambda^2(y)R^2R_{xx}(Rw_{xx} + w_x).$
At terminal time $t = T$, they become $ \vtl(T,x,y) = -\half\theta(y)D_1\vz(T,x)$ and $\vthl(T,x,y) = \half\rho\theta_1(y)BD_1^2\vz(T,x)$.

Note that the quantity $RR_{xx}(t,x;\aves)$ is bounded by a constant $K$. This is proved for $ (t,x;\aves)\in [0,T)\times \RR^+ \times \RR$ in Proposition~\ref{prop_ssh}, and guaranteed by Assumption~\ref{assump_U}\eqref{assump_Urisktolerance} for $t = T$, since by definition $R(T,x;\aves) = R(x)$. Therefore, the expectations related to the source terms in \eqref{eq_Eylprob} are sum of terms of the following form:
\begin{equation}\label{eq_opD}
\EE_{(t,x,y)}\left[\int_t^T h(Y_s)\MCD\vz(s,X_s^\pz) \ud s\right],
\end{equation}
where $h(y)$ is at most polynomially growing, and $\MCD\vz$ is one of the following:
$D_1^2\vz$, $D_1^3\vz$, $D_1^4\vz$, $D_1D_2D_1^2\vz$, $D_2D_1^2\vz$, $D_2D_1^3\vz$.
Applying Cauchy-Schwartz inequality, it becomes
$\EE_{(t,y)}^{1/2}\left[\int_t^T h^2(Y_s) \ud s\right]\EE_{(t,x,y)}^{1/2}\left[\int_t^T \left(\MCD\vz(s,X_s^\pz) \right)^2\ud s\right].$
The first part is uniformly bounded in $\eps$ since $Y_t$ admits bounded moments at any order (cf. Assumption~\ref{assump_valuefunc}\eqref{assump_valuefuncZmoment}). It remains to show the second part is also uniformly bounded in $\eps$.  The proof consists a repeated use of the concavity of $\vz$ and the results in Proposition \ref{prop_ssh} and Lemma \ref{lem_unibdd}. For the sake of simplicity, we shall only detail the proof when $\MCD\vz = D_1^2\vz$ and omit the rest. Since $\abs{D_1^2\vz} = \abs{RR_x\vz_x - R\vz_x} \leq (K_0+1)R\vz_x \leq (K_0+1) K_0x\vz_x \leq K_0(K_0+1)\vz$, we conclude
{\small\begin{multline*}
\EE_{(t,x,y)}\left[\int_t^T \left(D_1^2\vz(s,X_s^\pz, Y_s)\right)^2\ud s\right] \\ \leq K_0^2(K_0+1)^2 \EE_{(t,x,y)}\left[\int_t^T \left(\vz(s,X_s^\pz)\right)^2\ud s\right]
\end{multline*}}is uniformly bounded in $\eps$ by Lemma \ref{lem_unibdd}. Straightforward but tedious computations show that the rest terms in \eqref{eq_opD} are also bounded by multiples of $R\vz$, then the boundedness is again ensured by the relation $R(t,x;\aves) \leq K_0 x$, the concavity of $\vz$, and Lemma~\ref{lem_unibdd}.

The last two expectations in \eqref{eq_Eylprob} are treated similarly by using Assumption \ref{assump_U} \eqref{assump_Uiii} and the concavity of $U(x)$. Therefore we have shown that $\abs{\widetilde E(t,x,y)} \leq \widetilde C\eps$. By the inequality
$\abs{E(t,x,y)} \leq \widetilde C\eps + \eps \vtl(t,x,y) + \eps^{3/2}\vthl(t,x,y) \leq C\eps,$
we obtain the desired result.


\section{The Asymptotic Optimality of $\pz$}\label{sec_asymptoticoptimality}
We now show that the strategy $\pz$ defined in \eqref{eq_pz} asymptotically outperforms every family $\MCA_0^\eps\left[\pzt, \pot,\alpha\right]$ as precisely stated in our main Theorem \ref{Thm_optimality} in Section~\ref{sec_intro}.

For a fixed choice of $(\pzt$, $\pot)$ and positive $\alpha$, recall the definition of
$\MCA_0^\eps\left[\pzt, \pot,\alpha\right]$ in \eqref{def_class}.
Working with $\MCA_0^\eps$ is motivated by the following. The optimal control to problem \eqref{def_Vy}, whose existence is ensured by \cite{KrSc:03},
clearly depends on $\eps$. It is not known whether $\pi^\ast$ will converge as $\eps$ goes to zero. But if $\eps$ had a limit, say $\pzt$, it is then natural to consider a family of controls of the form $\pzt + \eps^\alpha\pot$
as the perturbation of the limit $\pzt$. We think the subset $\MCA_0^\eps$ is not so small comparing to the full one $\MCA^\eps$, as we only restrict $\alpha>0$, which allows for correction of any order in $\eps$.

\begin{assump}\label{assump_piregularity}
	For the triplet $(\pzt$, $\pot$, $\alpha)$, we require:
	\begin{enumerate}[(i)]
		\item The whole family of strategies $\{\pzt + \eps^\alpha \pot\}_{ \eps \leq 1} \in \MCA^\eps$ ;
		\item Let  $(\widetilde{X}_s^{t,x})_{t\leq s\leq T}$ be the solution to:
		{\small$\ud \widetilde X_s = \average{\mu(\cdot)\pzt(s,\widetilde X_s, \cdot)} \ud s + \sqrt{\average{\sigma^2(\cdot) \pzt(s,\widetilde X_s,\cdot)^2}} \ud W_s,$}
		starting at $x$ at time $t$. By (i), $\widetilde{X}_s^{t,x}\geq 0$. We further
		assume that it has full support $\RR^+$ for any $t<s\leq T$.
	\end{enumerate}
\end{assump}
\begin{rem}
	Part (ii) is motivated as follows. Consider
	$\ud \widehat X_s = \langle\mu\pz\rangle \ud s + \langle\sigma^2\pz^2\rangle^{\half} \ud W_s.$
Noticing that  $ \average{\mu(\cdot)\pz(t,x,\cdot)} = \overline{\lambda}^2 R(t,x;\aves)$,  $\sqrt{\average{\sigma^2(\cdot)\pz(t,x,\cdot)^2}} = \aves R(t,x;\aves)$, then $\widehat X_s$ can be interpreted as the optimal wealth process of the classical Merton problem with averaged Sharpe-ratio $\aves$.
	From \cite[Proposition 7]{KaZa:14}, one has	
	$\widehat X_s^{t,x} = H\left(H^{-1}(x,t,\aves) + \aves^2(s-t) + \aves(W_s-W_t), s, \aves\right),$
	where $H: \RR\times[0,T]\times \RR \to \RR^+$ solves the heat equation $H_t + \half \aves^2 H_{xx} = 0$,
	and is of full range in $x$. Consequently,  $\widehat X_s^{t,x}$ has full support $\RR^+$, and thus, it is natural to require that $\widetilde X_s^{t,x}$ has full support $\RR^+$.
\end{rem}

Denote by $\Vyt$ the value function associated to the trading strategy $\pi := \pzt + \eps^\alpha \pot \in \MCA_0^\eps\left[ \pzt, \pot, \alpha\right]$:
\begin{equation}\label{def_Vyt}
\Vyt(t,x,y) = \EE\left[U(X_T^\pi) \vert X_t^\pi = x, Y_t = y \right],
\end{equation}
where $X_t^\pi$ is the wealth process following the strategy $\pi\in \MCA_0^\eps$, and $Y_t$ is fast mean-reverting  with the same $\eps$.
The idea is to compare $\Vyt$ with $\Vyl$ defined in \eqref{def_Vyl}, for which a rigorous first order approximation $\vz+\sqrt{\eps}\vo$ has been established in Theorem \ref{Thm_one}. After finding the expansion of $\Vyt$, the comparison is done asymptotically in $\eps$ up to $\mathcal{O}(\sqrt{\eps})$. 



{\it Approximations of the Value Function $\Vyt$.} Denote by $\MCL$ the infinitesimal generator of the state processes $(X_t^\pi,Y_t)$:
$\MCL :=  \frac{1}{\eps}\MCL_0 + \frac{1}{2}\sigma^2(y)\left(\pzt + \eps^\alpha \pot \right)^2\partial_{xx} + \left(\pzt + \eps^\alpha\pot\right)\mu(y)\partial_x + \frac{1}{\sqrt{\eps}}\rho a(y)\sigma(y)\left(\pzt + \eps^\alpha\pot\right)\partial_{xy},$
then by the martingale property,  the value function $\Vyt$ defined in \eqref{def_Vyt} satisfies
\begin{equation}\label{eq_Vyt}
\partial_t \Vyt + \MCL \Vyt =0, \qquad \Vyt(T,x,y) = U(x).
\end{equation}
Motivate by the fact that the first order in the operator $\MCL$ is $\eps^\alpha$, we propose the following expansion form for $\Vyt$
$\Vyt = \vzt + \eps^\alpha \vat + \eps^{2\alpha}\vtat + \cdots + \eps^{n\alpha}\widetilde v^{n\alpha} + \sqrt{\eps}\,\vot +\cdots,$
where $n$ is the largest integer such that $n\alpha < 1/2$, and for the case $\alpha > 1/2$, $n$ is simply zero. In the derivation, we aim at identifying the zeroth order term $\vzt$ and the first non-zero term up to $\mathcal{O}(\sqrt\eps)$. Apparently, the term following $\vzt$ will depend on the value of $\alpha$.

To further simplify the notation, we decompose $\partial_t + \MCL$ according to different powers of $\eps$ as follows:
$\partial_t + \MCL = \frac{1}{\eps}\MCL_0 + \frac{1}{\sqrt{\eps}}\widetilde \MCL_1 + \widetilde \MCL_2 + \eps^\alpha \widetilde \MCL_3 + \eps^{2\alpha}\widetilde \MCL_4 + \eps^{\alpha-1/2}\widetilde \MCL_5,$
where the operators $\widetilde \MCL_i$ are defined by: $\widetilde \MCL_1 = \pzt\rho_1 a(y)\sigma(y)\partial_{xy}$, $\widetilde \MCL_2 =\partial_t +  \half\sigma^2(y)\left(\pzt\right)^2\partial_{xx} + \pzt\mu(y)\partial_x$, $\widetilde \MCL_3 = \sigma^2(y)\pzt\pot\partial_{xx} + \pot\mu(y)\partial_x$, $\widetilde \MCL_4 = \half\sigma^2(y)\left(\pot\right)^2\partial_{xx}$ and $\widetilde \MCL_5 = \pot\rho_1a(y)\sigma(y)\partial_{xy}$.

In all cases, we first collect terms of $\mathcal{O}(\eps^{\beta})$ in \eqref{eq_Vyt} with $ \beta \in [-1,0)$. Noticing that $\MCL_0$ and $\widetilde \MCL_1$ (also $\widetilde \MCL_5$ when $\alpha < 1/2$) take derivatives in $y$, we are able to make the choice that the approximation of $\Vyt$ up to $\mathcal{O}(\eps^{\beta'})$ is independent of $y$, for $\beta' < 1$. In the following derivation, this choice is made for every case, and consequently, we will not mention this again  and will start the argument by collecting terms of $\mathcal{O}(1)$. Different order of approximations are obtained depending  on $\pzt$ being identical to $\pz$ or not.

\subsubsection{Case $\pzt \equiv \pz$}
We first analyze the  case $\pzt \equiv \pz$, in which $\widetilde \MCL_1$ and $\widetilde \MCL_2$ coincide with $\MCL_1$ and $\MCL_2$, and $\widetilde \MCL_3 \vz = 0$. The terms of $\mathcal{O}(1)$ form a Poisson equation for $\vtt$
\begin{align}
\MCL_0 \vtt + \MCL_2 \vzt  = 0, \quad \vzt(T,x) = U(x).
\end{align}
For different values of $\alpha$, there might be extra terms which are eventually zero, thus are not included in the above equation: $\MCL_1 \vot$ (all cases), $\widetilde \MCL_5 \vzt $ when $\alpha = 1/2$, and $\widetilde \MCL_5 \widetilde v^{k\alpha}$ when $(k+1)\alpha = 1/2$. By the solvability condition, $\vzt$ solves \eqref{def_ltxpde}, which possesses a unique solution $\vz$. Therefore, we deduce
$\vzt \equiv \vz, \text{ and }\quad  \vtt \equiv \vtl.
$

\begin{enumerate}[(i)]
\item $\alpha = 1/2$.
We then collect terms of $\mathcal{O}(\eps^{1/2})$:
\begin{align}
 \MCL_0\vtht + \MCL_2\vot + \MCL_1\vtt + \widetilde \MCL_3 \vzt + \widetilde \MCL_5 \vot= 0.
\end{align}
This is a Poisson equation for $\vtht$, for which the solvability condition gives: $\vot$ satisfies \eqref{eq_vo}. Here we have used $\widetilde \MCL_3 \vzt = \widetilde \MCL_3 \vz = 0$, $\vtt = \vtl$ and $\widetilde \MCL_5\vot = 0$.
This equation is uniquely solved, one deduces
$
\vot = \vo,  \text{ and }\quad  \vtht \equiv \vthl.
$

\item $\alpha > 1/2$. Collecting terms of $\mathcal{O}(\sqrt{\eps})$ yields a Poisson equation for $\vtht$,
$
\MCL_0 \vtht + \MCL_2 \vot + \MCL_1 \vtt + \widetilde \MCL_5 \vzt = 0,
$
where the term $ \widetilde \MCL_5 \vzt$ only exists when $\alpha = 1$ (but anyway $\MCL_1\vot$ and $\widetilde \MCL_5\vzt$ disappear due to their independence of y). Arguments similar to the case $\alpha = 1/2$ give that
$
 \vot = \vo,  \text{and} \quad \vtht = \vthl.
$

\item $\alpha < 1/2$. The next order is $\eps^\alpha$,
\begin{equation}
\MCL_0\widetilde v^{\alpha+1} + \MCL_2\vat + \widetilde \MCL_3 \vzt + \MCL_1 \widetilde v^{\alpha + 1/2} + \widetilde \MCL_5 \vot =0.
\end{equation}
Again the last three terms disappear due to the fact $\widetilde \MCL_3 \vzt = \widetilde \MCL_3 \vz = 0$, and $\widetilde v^{\alpha + 1/2}$ and $\vot$'s independence of $y$. Then using solvability condition, $\vat$ solves
$
\Ltx(\aves) \vat(t,x;\aves) = 0, \quad \vat(T,x) = 0,
$
which only has the trivial solution $\vat \equiv 0$. Consequently, we need to identify the next non-zero term.

	$\pmb{1/4 < \alpha < 1/2}$. The next order is $\sqrt{\eps}$, which gives $
	\MCL_0 \vtht + \MCL_2 \vot + \MCL_1 \vtt =0. 
	$
	It coincides with \eqref{eq_vo} after using the solvability condition, and we deduce $\vot \equiv \vo$ and $\vtht \equiv \vthl$.
	
	$\pmb{\alpha = 1/4}$. The next order is $\sqrt\eps$, and the Poisson equation for $\vtht$ becomes
	$	\MCL_0 \vtht + \MCL_2 \vot + \MCL_1 \vtt +\widetilde \MCL_4 \vzt =0.
$
	The solvability condition reads $\Ltx(\aves)\vot - \half \rho BD_1^2\vz - \half \aves^2D_1\vz = 0$. Comparing this equation with \eqref{eq_vo} and using the concavity of $\vz$, one deduces $\vot \leq \vo$.
	
    $\pmb{\alpha < 1/4}$. The next order is $\eps^{2\alpha}$ since $2\alpha < 1/2$, and
    $\MCL_0 \widetilde v^{2\alpha+1} + \MCL_2\vtat + \MCL_1 \widetilde v^{2\alpha + 1/2} + \widetilde \MCL_3 \vat + \widetilde \MCL_4 \vzt + \widetilde \MCL_5 \widetilde v^{\alpha + 1/2} = 0, \quad \vtat(T,x) = 0.$
    The third, fourth and sixth terms varnish since $\vat \equiv 0$, and $\widetilde v^{2\alpha + 1/2}$ and $\widetilde v^{\alpha + 1/2}$ are independent of $y$. One has
    $\vtat_t + \half\aves^2 R^2\vtat_{xx}  + \aves R \vtat_x + \frac{1}{2}\average{\sigma^2(\cdot)\left(\pot(t,x,\cdot)\right)^2}\vz_{xx} = 0$,
    by the solvablility condition.
    Assuming that $\pot$ is not identically zero, we claim $\vtat <0$.


\end{enumerate}

%
\subsubsection{Case $\pzt \not \equiv \pz$}
In this case, after collecting terms of $\mathcal{O}(1)$, and using the solvability condition, one has the following PDE for $\vzt$:
		$\vzt_t + \frac{1}{2}\average{\sigma^2(\cdot)\pzt(t,x,\cdot)^2}\vzt_{xx} + \average{\pzt(t,x,\cdot)\mu(\cdot)}\vzt_x = 0.$
To compare $\vzt$ to $\vz$, we rewrite \eqref{eq_vz} in the same pattern:
		$\vz_t +\frac{1}{2}\average{\sigma^2(\cdot)\pzt(t,x,\cdot)^2}\vz_{xx} + \average{\pzt(t,x,\cdot)\mu(\cdot)}\vz_x - \frac{1}{2}\average{\sigma^2(\cdot)\left(\pzt -\pz \right)^2(t,x,\cdot)}\vz_{xx}  = 0$,
		via the relation $
		-\average{\sigma^2(y)(\pzt - \pz)\pz}\vz_{xx} = \average{(\pzt - \pz)\mu(y)}\vz_x.$
		Again by the strict concavity of $\vz$ and Feynman--Kac formula, we obtain $\vzt < \vz$.
		
%

To fully justify the above expansions, additional assumptions similar to \cite[Appendix C]{FoHu:16} are needed. They are technical uniform (in $\eps$) integrability conditions on the strategies $\MCA_0^\eps[\pzt, \pot, \alpha]$. For the sake of simplicity, we omit the conditions here and refer to \cite[Appendix C]{FoHu:16} for further details. Now we summarize the above derivation as follows.
\begin{prop}\label{prop_piaccuracy}
	Summary of the accuracy results:
	{\small	\begin{table}[H]
			\centering
			\caption{Accuray of approximations of $\Vyt$.}\label{tab_accuracy}
			\begin{tabular}{c|c|c|c}
				\hline\hline
				Case&Value of $\alpha$&Approximation&Accuracy\\ \hline
				&$\alpha \geq 1/2$ & $\vz + \sqrt{\eps}\vo$& $\MCO(\eps)$ \\ \cline{2-2}\cline{4-4}
				$\pzt \equiv \pz$&$1/4 < \alpha < 1/2$& & $\MCO(\eps^{2\alpha})$   \\ \cline{2-4}
				&$\alpha = 1/4$& $\vz + \sqrt{\eps}\vot$& $\MCO(\eps^{3/4})$ \\\cline{2-4}
				&$\alpha < 1/4$& $\vz + \eps^{2\alpha}\vtat$&$\MCO(\eps^{3\alpha \wedge (1/2)})$ \\  \hline
				
					$\pzt \not\equiv \pz$& all& $\vzt$&$ \MCO(\eps^{\alpha \wedge (1/2)})$ \\ \hline\hline
			\end{tabular}
		\end{table}
	}
	\noindent where the accuracy column gives the order of the difference between $\Vyt$ and its approximation. Moreover, when $\pzt \equiv \pz$, we have the relation $\vot \leq \vo$ if $\alpha = 1/4$, and $\vtat <0$ if $\alpha < 1/4$; while if $\pzt \not \equiv \pz$, then $\vzt < \vz$.
\end{prop}


{\it Asymptotic Optimality: Proof of Theorem~\ref{Thm_optimality}.} We now give the proof of Theorem~\ref{Thm_optimality}, via comparing the first order approximation $\vz + \sqrt\eps{\vo}$ of $\Vyl$ obtained in Theorem~\ref{Thm_one}, and the one of $\Vyt$ summarized in Tab.~\ref{tab_accuracy}.

In the case that the approximation of $\Vyt$ is $\vz + \sqrt{\eps}\vo$, the limit is easily verified to be zero. When the approximation of $\Vyt$ is $\vz + \sqrt{\eps}\vot$, the limit $\ell$ is non-positive but stay finite, by the fact $\vot \leq \vo$. If $\pzt \equiv \pz $ and $\alpha < 1/4$, 
the limit $\ell$ is computed as
$
\ell = \lim_{\eps \to 0}\bigl(\eps^{2\alpha}\vtat - \sqrt{\eps}\vo + \MCO(\eps^{3\alpha \wedge 1/2})\bigr)/{\sqrt{\eps}} = -\infty,
$
since $\vtat<0$. The similar arguments also apply to the case $\pzt \not \equiv \pz$, and lead to $\ell = -\infty$. Thus we complete the proof. In fact, this limit can be understood according to the following four cases:
	\begin{enumerate}[(i)]
		\item$\pzt \equiv \pz$ and $\ell = 0$: $\Vyt = \Vyl + o(\sqrt{\eps})$;
		\item $\pzt \equiv \pz$ and $-\infty < \ell <0$: $\Vyt = \Vyl + \MCO(\sqrt{\eps})$ with $\MCO(\sqrt{\eps}) <0$;
		\item $\pzt \equiv \pz$ and $\ell = -\infty$: $\Vyt = \Vyl + \MCO(\eps^{2\alpha})$ with $\MCO(\eps^{2\alpha})<0$ and $2\alpha < 1/2$;
		\item $\pzt \not \equiv \pz$: $\displaystyle
		\lim_{\eps \to 0} \Vyt(t,x,z)< \lim_{\eps \to 0}  \Vyl(t,x,z).$
	\end{enumerate}

\bibliographystyle{plain}
\bibliography{Reference}

\end{document}